\documentclass{article}
\usepackage{kr}
\usepackage{times}
\usepackage{soul}
\usepackage{cancel}
\usepackage[small]{caption}
\usepackage{booktabs}
\usepackage{algorithm}
\usepackage{algorithmic}

\usepackage{eucal}

\usepackage[utf8]{inputenc} 
\usepackage[T1]{fontenc}    
\usepackage{hyperref}       
\usepackage{url}            
\usepackage{booktabs}       
\usepackage{amsfonts}       
\usepackage{nicefrac}       
\usepackage{microtype}      
\usepackage{lipsum}
\usepackage{fancyhdr}       
\usepackage{graphicx}       
\graphicspath{{media/}}     
\usepackage{wrapfig}        

\usepackage{complexity}
\mathchardef\mhyphen="2D 
\usepackage{arydshln}
\usepackage{xcolor}
\usepackage{enumitem}
\usepackage{bigstrut}
\usepackage{comment}
\usepackage{tikz}
\usepackage{adjustbox}
\usetikzlibrary{positioning} 
\usepackage{stmaryrd}
\usepackage{amssymb,amsmath,amsthm,mdsymbol}

\usepackage[switch,mathlines]{lineno}
\usepackage{etoolbox}

\newcommand*\linenomathpatch[1]{%
  \cspreto{#1}{\linenomath}%
  \cspreto{#1*}{\linenomath}%
  \csappto{end#1}{\endlinenomath}%
  \csappto{end#1*}{\endlinenomath}%
}
\newcommand*\linenomathpatchAMS[1]{%
  \cspreto{#1}{\linenomathAMS}%
  \cspreto{#1*}{\linenomathAMS}%
  \csappto{end#1}{\endlinenomath}%
  \csappto{end#1*}{\endlinenomath}%
}

\expandafter\ifx\linenomath\linenomathWithnumbers
  \let\linenomathAMS\linenomathWithnumbers
  \patchcmd\linenomathAMS{\advance\postdisplaypenalty\linenopenalty}{}{}{}
\else
  \let\linenomathAMS\linenomathNonumbers
\fi

\linenomathpatch{equation}
\linenomathpatchAMS{gather}
\linenomathpatchAMS{multline}
\linenomathpatchAMS{align}
\linenomathpatchAMS{alignat}
\linenomathpatchAMS{flalign}

\makeatletter
\patchcmd{\mmeasure@}{\measuring@true}{
  \measuring@true
  \ifnum-\linenopenaltypar>\interdisplaylinepenalty
    \advance\interdisplaylinepenalty-\linenopenalty
  \fi
  }{}{}
\makeatother

\newtheorem{theorem}{Theorem}
\newtheorem{corollary}[theorem]{Corollary}

\newtheorem{proposition}[theorem]{Proposition}

\theoremstyle{definition} 
\newtheorem{definition}[theorem]{Definition}
\newtheorem{example}[theorem]{Example}

\newcommand{\eval}[3]{\llbracket{#1}\rrbracket_{{#2},{#3}}}
\newcommand{\evalX}[2]{\llbracket{#1}\rrbracket_{{#2}}}
\newcommand{\mA}{\mathfrak{A}}
\newcommand{\mB}{\mathfrak{B}}
\newcommand{\mC}{\mathfrak{C}}
\newcommand{\lit}[1]{\mathsf{Lit}_{#1}}

\newcommand{\dep}[1]{=\!\!(#1)}
\newcommand{\Dom}{\mathrm{Dom}}
\newcommand{\Ran}{\mathrm{Ran}}
\newcommand{\lite}[1]{\mathrm{lit}\text{-}{#1}}
\newcommand{\liteprov}[1]{\mathrm{prov}\text{-}{#1}}
\newcommand{\X}{\mathbb{X}}

\newcommand{\proj}[2]{{#1}_{\upharpoonright {#2}}}
\newcommand{\sub}{\subseteq}

\newcommand{\dfn}{\coloneqq}
\newcommand{\Var}{\textrm{Var}}
\newcommand{\Fr}{\textrm{Fr}}
\newcommand{\Pow}{\mathcal{P}} 
\newcommand{\pind}{
\perp}

\newcommand{\pci}[3]{#2~\!\!\pind_{#1}\!\!~#3}
\newcommand{\pmi}[2]{#1~\!\!\pind\!\!~#2}
\newcommand{\nonpmi}[2]{#1~\!\!\not\pind\!\!~#2}

\newcommand{\zt}{\bot?}
\newcommand{\nzt}{\cancel{\zt}}

\newcommand{\simdif}[2]{#1\triangle\!#2}

\newcommand{\fofull}{\FO(=,\neq, \leq, \not\leq)}
\newcommand{\fopos}{\FO(=, \nzt,\leq)}
\newcommand{\foneg}{\FO(\zt, \neq, \not\leq)}

\newcommand{\nnf}{\mathrm{nnf}}
\newcommand{\ar}{\mathrm{ar}}
\newcommand{\as}{\textrm{As}}
\newcommand{\KX}{\mathbb{X}}
\newcommand{\KY}{\mathbb{Y}}
\newcommand{\KZ}{\mathbb{Z}}

\newcommand{\SB}{\mathbb{B}}
\newcommand{\support}{\mathrm{Sup}}

\newcommand{\cmap}[1]{\xi_{#1}}

\title{Unified Foundations of Team Semantics via Semirings
}

\author{%
Timon Barlag$^1$\and
Miika Hannula$^2$\and
Juha Kontinen$^{2}$\and
Nina Pardal$^3$\and
Jonni Virtema$^3$ \\
\affiliations
$^1$Leibniz Universit\"at Hannover\\
$^2$University of Helsinki\\
$^3$University of Sheffield\\
\emails
barlag@thi.uni-hannover.de,
\{miika.hannula, juha.kontinen\}@helsinki.fi,\\
\{n.pardal, j.t.virtema\}@sheffield.ac.uk.
}

\begin{document}

\maketitle

\begin{abstract}
Semiring semantics for first-order logic provides a way to trace how facts represented by a model are used to deduce satisfaction of a formula.
Team semantics is a framework for studying logics of dependence and independence in diverse contexts such as databases, quantum mechanics, and statistics by extending first-order logic with atoms that describe dependencies between variables. Combining these two, we propose a unifying approach for analysing the concepts of dependence and independence via a novel semiring team semantics, which subsumes all the previously considered variants for first-order team semantics.
%
In particular, we study the preservation  of satisfaction of dependencies and formulae between different semirings. In addition we create links to reasoning tasks such as provenance, counting, and repairs.

\end{abstract}


\section{Introduction}
Team semantics offers a 
logical framework to study important concepts that arise in the presence of plurality of data such as dependence and independence. 
The birth of the area can be traced back to the introduction of dependence logic in \cite{vaananen07}.
During the past decade, the expressivity and complexity theoretical aspects of logics in team semantics have been actively studied.
Fascinating connections have been drawn to areas such as of database theory \cite{HannulaKV20,HannulaK16}, verification \cite{GutsfeldMOV22}, real valued computation \cite{abs-2003-00644}, and quantum foundations \cite{AlbertG22,abramsky2021team}.
The study has focused on logics in the first-order, propositional and modal team semantics, and more recently also in the multiset \cite{DurandHKMV18,GradelW22} and probabilistic settings \cite{HKMV18}. 
Prior to this work, these adaptations of team semantics have been studied in isolation from one another.
%

Data provenance provides means to describe the origins of data, allowing to give information about the witnesses to a query, or determining how a certain output is derived. 
Provenance semirings were introduced in \cite{GreenKT07} to devise a general framework that allows to uniformly treat extensions of positive relational algebra, where the tuples have annotations that reflect very diverse information. Some motivating examples of said relations come from incomplete and probabilistic databases, and bag semantics. 
This semiring framework captures a notion of data provenance called \emph{how-provenance}, where the semiring operations essentially capture how each output is produced from the source. 
Following this framework, semiring semantics for full first-order logic ($\FO$) were developed in \cite{gradelarxiv17}. 
The semiring semantics for $\FO$ refines, in particular, the classical Boolean semantics by allowing formulae to be evaluated as values from a semiring. This allows for example counting proof trees, or winning strategies in the model checking game for $\mA$ and $\phi$. 




In databases, dependencies are applied as integrity constraints (ICs) that specify sets of rules that the database needs to satisfy. Formal analysis of the rules is facilitated by viewing them as $\FO$ sentences that usually follow certain syntactic patterns.
This approach is sometimes inadequate because query languages such as SQL operate with multisets (i.e., \emph{bags}) of tuples instead of sets. Recently, 
\cite{ChuMRCS18} have formulated ICs, such as keys and foreign keys, over semirings to study SQL query equivalence. In probability theory, conditional independence has widespread applications; for instance, assumptions about conditional independence can simplify computations of joint probabilities of variables. It is known that dependency notions in database theory and probability theory are interlinked because many (but not all) such concepts can be rewritten in terms of information-theoretic measures such as conditional entropy and conditional mutual information \cite{Lee87,5075885}. However, we are not aware of previous works that use semirings to the same effect.



Similar to extending first-order logic with counting,
 we extend the semiring semantics of $\FO$ \cite{gradelarxiv17} with the ability of comparing the semiring  values of first-order formulae.
Using this formalism we define concepts such as dependence and independence in a way that encompasses prior interpretations. The proposed formalism also provides a robust framework for studying the preservation of satisfaction and entailment for dependence statements when moving from one semiring to another. Such preservation results have previously been studied between database and probability theory \cite{2242326,GyssensNG14,DurandHKMV18,KenigS22,Malvestuto86,Malvestuto92,WongBW00}.  
Furthermore, we propose a unified approach to team semantics that involves annotating the elements of a team with elements from an arbitrary semiring.
 By doing so, the original team semantics and its quantitative variants 
can be recovered by choosing a suitable concrete semiring.
The conversion to semiring team semantics enables provenance analysis and other reasoning tasks to be performed for the first time for expressive team-based logics. 
\section{Preliminaries}
We fix a countably infinite set $\Var$ of \emph{variables}. 
We use $\mA,\mB,\mC,\dots$ to denote \emph{first-order structures}, and write $A,B,C,\dots$ for their \emph{domains}.
An \emph{assignment} (of $\mA$) is a function $s$ that maps a finite set $D\sub \Var$ of variables to some values (in $A$). 
We call $D$ the \emph{domain} of $s$, written $\Dom(s)$.
For a variable $x$ and a value $a$, we write $s[a/x]$ for the assignment with domain $\Dom(s)\cup\{x\}$ which maps $x$ to $a$ and otherwise agrees with $s$.

%

A \emph{team} $X$ is a finite set of assignments $s$ with a shared domain $D$. We call $D$ the \emph{domain} of $X$, written $\Dom(X)$.
Given a first-order structure $\mA$, we say that $X$ is a \emph{team of} $\mA$, if $A$ subsumes the ranges of each $s\in X$.
%
%
%
Moving from single assignments to sets of assignments enables us to interpret 
dependency statements between variables:

A \emph{conditional independence atom} is an expression of the form $\pci{\vec{x}}{\vec{y}}{\vec{z}}$, where $\vec{x},\vec{y},\vec{z}$ are variable sequences (not necessarily of the same length).
A team $X$ \emph{satisfies} $\pci{\vec{x}}{\vec{y}}{\vec{z}}$, written $X\models \pci{\vec{x}}{\vec{y}}{\vec{z}}$, if for all $s,s'\in X$ such that $s(\vec{x})=s'(\vec{x})$ there exists $s'' \in X$ such that $s(\vec{x}\vec{y})=s''(\vec{x}\vec{y})$ and $s'(\vec{z})=s''(\vec{z})$. A \emph{pure independence atom} is an expression of the form $\pmi{\vec{x}}{\vec{y}}$, defined as $\pci{\emptyset}{\vec{x}}{\vec{y}}$. A team $X$ satisfies $\pmi{\vec{x}}{\vec{y}}$, if for all $s,s'\in X$ there exists $s''\in X$ such that $s(\vec{x})=s''(\vec{x})$ and $s'(\vec{y})=s''(\vec{y})$.

A \emph{dependence atom} is an expression of the form $\dep{\vec{x},\vec{y}}$, where $\vec{x}$ and $\vec{y}$ are variable sequences. 
A team $X$ \emph{satisfies} $\dep{\vec{x},\vec{y}}$, if for all $s,s'\in X$, $s(\vec{x})=s'(\vec{x})$ implies $s(\vec{y})=s'(\vec{y})$.

An \emph{inclusion atom} is an expression of the form $\vec{x} \subseteq \vec{y}$, where $\vec{x}$ and $\vec{y}$ are variables sequences of the same length. A team $X$ \emph{satisfies} $\vec{x} \subseteq \vec{y}$, if for all $s \in X$ there is $s'\in X$ such that $s(\vec{x}) = s'(\vec{y})$.

In the probabilistic team semantics setting, $\pci{\vec{x}}{\vec{y}}{\vec{z}}$ is given the usual meaning of conditional independence in probability theory. Furthermore, the probabilistic interpretation of $\vec{x} \subseteq \vec{y}$ states that the marginal distributions of  $\vec{x}$ and $\vec{y}$ are identical.

If $\alpha$ is any atom for which satisfaction by a team $X$ is defined, we extend this definition to first-order structures $\mA$ by saying that $X$ \emph{satisfies} $\alpha$ \emph{under} $\mA$, written $\mA \models_X \alpha$, if $X$ satisfies $\alpha$.

\section{Semiring Perspective on Dependencies}\label{sect:semiring-perspective}
In this section we consider generalising teams and dependencies using semirings.

\subsection{Semirings}
We start by briefly reviewing semiring related concepts that are necessary for the present paper.  
  
\begin{definition}[Semiring]
A \emph{semiring} is a tuple $K=(K,+,\cdot,0,1)$, where $+$ and $\cdot$ are binary operations on $K$, $(K,+,0)$ is a commutative monoid with identity element $0$, $(K,\cdot ,1)$ is a monoid with identity element $1$, $\cdot$ left and right distributes over $+$, and $x \cdot 0 =0= 0\cdot x$ for all $x \in K$.
 $K$ is called \emph{commutative} if $(K,\cdot ,1)$ is a commutative monoid.
 As usual, we often
write $ab$ instead of $a \cdot b$. 
\end{definition}
That is, semirings are rings which need not have additive inverses. We focus on the listed semirings that 
 encapsulate the set, multiset, and distribution based team semantics:

\begin{itemize}
	\item The \emph{Boolean semiring} $\mathbb{B}=(\mathbb{B},\lor,\land,0,1)$ models logical truth and is formed from the two-element Boolean algebra. It is the simplest example of a semiring that is not a ring.
	\item The \emph{probability semiring} $\mathbb{R}_{\geq 0}=(\mathbb{R}_{\geq 0},+,\cdot,0,1)$ consists of the non-negative reals with standard
	addition and multiplication.
	\item The \emph{semiring of natural numbers} $\mathbb{N}=(\mathbb{N},+,\cdot,0,1)$ consists of natural numbers with their usual operations.
 \end{itemize}
 Other examples include the semiring of multivariate polynomials $\mathbb{N}[X]=(\mathbb{N}[X],+,\cdot,0,1)$ which is the free commutative semirings generated by the indeterminates in $X$, 
 the
 \emph{tropical semiring} $\mathbb{T} = (\mathbb{R}\cup\{\infty\}, \min, +, \infty, 0)$ which consists of the reals expanded with infinity
	and has $\min$ and $+$ respectively plugged in for addition and multiplication, and
%
the Lukasiewicz semiring $\mathbb{L} = ([0,1], \max, \cdot , 0, 1)$, used in multivalued logic, which endows the 
	unit interval with $\max$ addition and multiplication $a \cdot b \dfn \min(0, a+b-1)$.

%
Let $\leq$ be a partial order. 
A binary operator $*$ is said to be \emph{monotone under $\leq$} if $a\leq b$ and $a'\leq b'$ implies $a*a' \leq b*b'$. 
A \emph{partially ordered semiring} is a tuple $K=(K,+,\cdot,0,1, \leq)$, where $(K,+,\cdot,0,1)$ is a semiring, $(K,\leq)$ is a partially ordered set, and $+,\cdot$ are monotone under $\leq$. Given a semiring ${K}=(K,+,\cdot,0,1)$, define a binary relation $\leq_{K}$ on $K$ as $a \leq_{K} b$ if $\exists c : a+c = b$.  
This relation is a preorder; meaning it is reflexive and transitive. If $\leq_{K}$ is also antisymmetric, it is a partial order, called the \emph{natural order} of ${K}$, and ${K}$ is said to be \emph{naturally ordered}. 
In this case, ${K}$ endowed with its natural order is a partially ordered semiring.

If a semiring $K$ satisfies $ab=0$ for some $a,b\in K$ where $a\neq 0 \neq b$, 
we say that $K$ has \emph{divisors of $0$}. On the other hand, 
a semiring $K$ is considered \emph{$+$-positive} if $a+b=0$ implies 
that $a=b=0$. If a semiring is both $+$-positive and has no divisors of $0$, 
it is referred to as \emph{positive}.
For example, the modulo two integer semiring $\mathbb{Z}_2$ 
is not positive since it is not $+$-positive (even though it has no divisors of $0$). 
Conversely, an example of a semiring with divisors of $0$ is $\mathbb{Z}_4$.
We can also examine the positivity of $K$ by looking at its 
\emph{characteristic mapping}, which is defined as the function 
$\cmap{K}:K\to \SB$ such that
\[
\cmap{K}(a) = 
\begin{cases}
1 &\text{ if }a \neq 0,\\
0 &\text{ if }a=0.
\end{cases}
\] 

\begin{proposition}[Proposition 6, \cite{gradelarxiv17}]
	A semiring $K$ is positive if and only if its characteristic mapping is a homomorphism. 
	\end{proposition}
	
In particular, note that
	the Boolean semiring $\mathbb{B}$, the probability semiring  $\mathbb{R}_{\geq 0}$, and the semiring of natural numbers $\mathbb{N}$ are positive and naturally ordered.

\subsection{$K$-teams}
%
Given a semiring $K$, the concept of a $K$-team is obtained by labeling each assignment of a team with an element from $K$. 
%
If a $D$ is a set of variables and $A$ a set, we define $\as(D,A)$ as the set all assignments $s:D\to A$. 
\begin{definition}[$K$-team]
A $K$\emph{-team} is a
function $\KX \colon \as(D,A)\to K$, where
$ K=(K,+,\cdot,0,1)$ is a (commutative) semiring, $D$ is a set of variables, and A is a set.
The \emph{support} of $\KX$ is defined as $\support(\KX)\dfn \{ s \mid \KX(s)\neq 0\}$.
Provided that we have an ordering on $K$, we say that $\KX$ is a \emph{subteam} of $\KY$ if $\KX(s) \leq \KY(s)$ for every $s \in \support(\KX)$.

\end{definition}    

We can now reconceptualise the notion of a team as a $K$-team by associating 
each possible assignment with either a $1$ or $0$ label depending 
on whether or not it belongs to the team. 
 When dealing with probabilistic teams, each assignment is labeled 
 with a non-negative 
 real number that can be interpreted as a probability by scaling the sum of these 
 labels to one. For multiteams, each assignment is assigned a natural number. 
 Fig. \ref{tab:three} provides an illustration of these concepts.

\begin{figure}[h!]
	\caption{$K$-teams $\KX_i \colon  \as(D,A)\to K$, $D=\{x,y\}$, $A=\{a,b\}$, 
	representing a team, a multiteam, and a probabilistic team.}\label{tab:three}
	\centering
	\begin{tabular}{@{}cc:c@{}}
	\multicolumn{3}{c}{$K= \mathbb{B}$} \\
	\toprule
	$x$ & $y$ & $\KX_1(s)$ \\
	\midrule
	a & a & 1 \\
	a & b & 1 \\
	b & a & 0 \\
	b & b & 0 \\
	\bottomrule
	\end{tabular}
	\hspace{1cm}
	\begin{tabular}{@{}cc:c@{}}
	\multicolumn{3}{c}{$ K = \mathbb{N}$} \\
	\toprule
	$x$ & $y$ & $\KX_2(s)$ \\
	\midrule
	a & a & 2 \\
	a & b & 0 \\
	b & a & 0 \\
	b & b & 5 \\
	\bottomrule
	\end{tabular}
	\hspace{1cm}
	\begin{tabular}{@{}cc:c@{}}
	\multicolumn{3}{c}{$K = \mathbb{R}_{\geq 0}$} \\
	\toprule
	$x$ & $y$ & $\KX_3(s)$ \\
	\midrule
	a & a & 1/4 \\
	a & b & 3/4 \\
	b & a & 0 \\
	b & b & 0 \\
	\bottomrule
	\end{tabular}
\end{figure}

\subsection{Dependencies over $K$-teams: A Prologue}
Our goal is to find a common language for expressing concepts 
such as dependence and independence in different $K$-teams.
Referring back to the preliminaries section, the reader may observe that
the fundamental dependency concepts in team semantics can 
be formalised using the language of first-order logic.  
This approach, however, becomes insufficient when dealing with multisets
 or probability distributions. For example, the concept of 
 independence between two random variables involves counting, 
 which is beyond the capabilities of first-order logic; this appears to be the case in the semiring 
 context as well  \cite{abs-2203-03425}. Therefore, we explore 
   extensions of first-order logic to overcome such limitations. The following example hints at the direction 
 we will take.
 
  \begin{example}\label{ex:logicaldef}
 We aim to find a common logical expression underlying both relational
 and probabilistic interpretations of conditional independence.
 To this end, fix a conditional independence atom $\pci{x}{y}{z}$ over variables $x,y,z$.
 
 In the relational context, viewing a team $X$ with domain $\{x,y,z\}$ as a ternary relation $R=\{(s(x),s(y),s(z))\mid s \in X\}$, 
 we observe that
  $X$ satisfies $\pci{x}{y}{z}$ if and only if 
 $R$ satisfies the first-order sentence
 \begin{equation}\label{eq:reldef}
	 \forall abcde (R(a,b,c)\land R(a,d,e) \to R(a,b,e)).
 \end{equation}
 
 Moving to the probability context, two random variables $y$ and $z$ are 
 conditionally independent given a random variable $x$ if and only if for
 all values $a,b,c$,  
 \begin{equation}\label{eq:probdef}
	\hspace{-2mm}P(y=b\mid x=a)\cdot P(z=c\mid x=a) = P(yz=bc \mid x=a).
\end{equation} 

Our strategy is to 
transform \eqref{eq:probdef} into a ``logical'' sentence similar to \eqref{eq:reldef}.
First, we remove conditional probabilities to obtain from \eqref{eq:probdef} the equation
\begin{equation*}
 P(xy=ab)\cdot P(xz=ac) = P(xyz=abc)\cdot P(x=a).
\end{equation*} 
Next, we model the probability distribution with a ternary function $R$ 
mapping value triples $(a,b,c)$ to the probabilities $P(xyz=abc)$, and rewrite a marginal probability $P(xy=ab)$ as 
the sum of probabilities $\sum_{c} R(a,b,c)$, arriving at
\begin{equation*}
	\sum_{c} R(a,b,c) \cdot \sum_{b} R(a,b,c) = R(a,b,c)\cdot \sum_{b,c} R(a,b,c).
\end{equation*}
By interpreting multiplication as conjunction and aggregate summation as existential quantification,
 and adding the universal quantification of triples,  we arrive at the expression 
\begin{multline}\label{eq:probdef3}
	\forall abc \big((\exists c R(a,b,c) \land \exists b R(a,b,c)) \\ = (R(a,b,c)\land \exists bc R(a,b,c))\big).
\end{multline}
This expression can be viewed as a ``logical'' sentence defining probabilistic independence. 
Note that it involves an equality statement between two formulae and is thus not a well-formed first-order sentence. 
However, if we replace the equality symbol $=$ with the logical equivalence symbol $\leftrightarrow$,
we obtain a first-order sentence
which, after removing logical redundancies, transforms into
\begin{equation*}
	\forall abc \big((\exists c R(a,b,c) \land \exists b R(a,b,c)) \to (R(a,b,c))\big).
\end{equation*}
By renaming the existentially quantified variables, and dragging them in front of the quantifier-free part, we obtain precisely the first-order sentence \eqref{eq:reldef}
that we used to define relational conditional independence.
 \end{example}
 
 Based on the example, it appears that logical statements formulated in the manner of \eqref{eq:probdef3} can integrate diverse expressions of dependency concepts. To give more depth to this idea,
 we will dedicate the next section to the interpretation of first-order logic and its extension with equality statements between formulae within the semiring context.
 Following this, we will revisit  the concept of dependencies over $K$-teams.
  
 \section{First-order logic with formula equality}\label{sect:foeq}


We first review $K$-interpretations for first-order formulae from \cite{gradelarxiv17}.
%
From now on, we consider only commutative semirings. This is necessary to properly interpret quantifiers within this context.
\subsection{First-order interpretations}

Fix a relational vocabulary $\tau= \{R, S, T, ...\}$. We denote by $\ar(R)$ the $\emph{arity}$ of a relational symbol $R$.
A \emph{relational atom} (resp. a \emph{negated relational atom}) is an expression of the form $R(\vec{x})$ (resp. $\neg R(\vec{x})$) where $\vec{x}$ is a sequence of variables of length $\ar(R)$. An \emph{equality atom} (resp. \emph{negated equality atom}) is an expression of the form $x=y$ where $x$ and $y$ are variables. An atom or a negated atom is called a \emph{literal}.
\emph{First-order formulae} are the expressions formed by closing atoms
 by quantifiers
$\exists,\forall$ and connectives $\land,\lor,\neg$ in the usual way. 
%
We use $\phi \to \psi$ as a shorthand for $\neg \phi\lor \psi$, and $\phi \leftrightarrow \psi$ as a shorthand for $(\phi \to \psi) \land (\psi \to \phi)$. 
The \emph{set of free variables} $\Fr(\theta)$ of an $\FO$ formula $\theta$ is defined in the usual way. 
We also write $\nnf$ for the standard negation normal form transformation of first-order formulae.

Let $A$ be a set. An expression of the form $R(\vec{a})$ (resp. $\neg R(\vec{a})$), where $\vec{a}\in A^{\ar(R)}$, is 
called a \emph{fact} (resp. \emph{negated fact}) over $A$.
The set of \emph{literals} over $A$, denoted by $\lit{A}$, is the set comprising all facts and negated facts over $A$.
\begin{definition}[\cite{gradelarxiv17}]\label{def:fointer}
Fix a semiring ${K}=(K,+,\cdot, 0, 1)$. A $K$-interpretation is a mapping $\pi\colon \lit{A} \to K$. Given variable assignments $s\colon \Var \to A$, it is extended to $\FO$ formulae as follows:
\begin{align*}
     \eval{R(\vec{x})}{\pi}{s}&= \pi(R(s(\vec{x})))  
     &  \eval{\phi \land \psi}{\pi}{s}  & = \eval{\phi}{\pi}{s} \cdot \eval{ \psi}{\pi}{s}   \\
     \eval{\neg R(\vec{x})}{\pi}{s}&= \pi(\neg R(s(\vec{x})))
     &   \eval{\phi \lor \psi}{\pi}{s} & = \eval{\phi}{\pi}{s} + \eval{ \psi}{\pi}{s}  \\
    \eval{\forall x\phi}{\pi}{s} & =    \prod_{a \in A} \eval{\phi }{\pi}{s[a/x]} 
    &  \eval{\exists x\phi}{\pi}{s} & =    \sum_{a \in A} \eval{\phi }{\pi}{s[a/x]}\\
     \eval{\neg \phi}{\pi}{s} & = \eval{\nnf(\neg\phi)}{\pi}{s} &
       \eval{x * y}{\pi}{s} &=  
     \begin{cases}
     1 &\hspace{-2mm}\text{if }s(x)*s(y)\\
     0 &\hspace{-2mm}\text{otherwise},
     \end{cases}
  \end{align*}
where $*\in \{=, \neq\}$. If $\Fr(\phi)$ is empty, then $\phi$ is called a \emph{sentence}. For sentences $\phi$, we write
$\evalX{\phi}{\pi}$ as a shorthand for $\eval{\phi}{\pi}{s_\emptyset}$, where $s_\emptyset$ is the empty assignment.
\end{definition}
A $K$-interpretation $\pi$ is called \emph{model-defining} \cite{gradelarxiv17} if for all facts $R(\vec{a})$ it holds that exactly one of $R(\vec{a})$ and $\neg R(\vec{a})$ is mapped to $0$ by $\pi$, while the other is mapped to a value different from $0$. 

The compositional interpretation entails that semiring homomorphisms extend to formula interpretations. This property will be used in this paper
to analyse mutual relationships between different interpretations of dependency concepts.

\begin{proposition}[\cite{gradelarxiv17}]
Let $h$ be a semiring homomorphism from $K_1$ to $K_2$, and let $\pi_1:\lit{A}\to K_1$ and $\pi_2:\lit{A} \to K_2$ be interpretations such that $h\circ \pi_1 = \pi_2$. Then, $h(\evalX{\phi}{\pi_1}) = \evalX{\phi}{\pi_2}$ for every first-order sentence $\phi$.
 \end{proposition}


Let $\mA$ be a standard first-order structure over $\tau$, and let $\mathbb{B}$ be the Boolean semiring. The interpretation $\pi$ that maps relational facts $R(\vec{a})$ (resp. negated relational facts $\neg R(\vec{a})$) to $1$ (resp. $0$) if $\vec{a} \in R^{\mA}$, and otherwise to $0$ (resp. $1$), is called the \emph{canonical truth interpretation} of $\mA$,  denoted $\pi_\mA$.
\begin{proposition}[\cite{gradelarxiv17}]\label{prop:can}
Let $\phi$ be a first-order sentence, 
and $\mA$ a structure. Then $\mA\models \phi$ if and only if $\evalX{\phi}{\pi_\mA} =1 $.
\end{proposition}

\subsection{Formula (in)equality}\label{Sec_foieq} 
To express dependencies logically in a general semiring context, we introduce equality and inequality over FO formulae. 
Given $\phi,\psi \in \FO$, 
we extend the $K$-interpretation as follows:
\[
       \eval{\phi * \psi}{\pi}{s} = 
     \begin{cases}
     1 &\text{ if }\eval{\phi}{\pi}{s} *\eval{\psi}{\pi}{s} \\
     0 & \text{ otherwise},
     \end{cases}
\]
where $* \in \{=,\neq, \leq, \not\leq \}$.
For the (negated) formula inequality, we assume $(K,\leq)$ is a partially ordered semiring.
We write $\zt$ and $\nzt$ to denote the formula equalities of the form $\phi = \bot$ and $\phi \neq \bot$, respectively.
For $C \subseteq \{\zt, \nzt, =,\neq, \leq, \not\leq \}$, we let $\FO(C)$ denote the extension of the logic of Gr\"adel and Tannen with the formula equalities and inequalities in $C$ occuring positively (i.e. in the scope of even number of negations) and without nesting.
The \emph{set of free variables} for a formula of the form $\phi *\psi$, $*\in \{=,\neq, \leq, \not\leq \}$, is defined as $\Fr(\phi * \psi) = \Fr(\phi) \cup \Fr( \psi)$.

We extend $\nnf$ to formula (in)equalities by setting $\nnf (\neg(\phi = \psi))\dfn (\nnf(\phi) \neq \nnf(\psi)) $ and $\nnf(\neg(\phi \leq \psi))\dfn (\nnf(\phi) \not\leq \nnf(\psi)) $.
We can then extend the use of shorthands $\phi \to \psi$ and  $\phi \leftrightarrow \psi$ for FO with formula (in)equalities. 


\begin{proposition}\label{prop:posmodeldef}
Let $K$ be a positive semiring, and let $\pi$ be a model-defining $K$-interpretation. Let $\phi$ be a formula of $\fofull$. Then, $\pi(\phi)=0$ if and only if $\pi(\neg\phi)\neq 0$.
\end{proposition}
\begin{proof}
The proof is by structural induction. If $\phi$ is an atom, the statement follows by the assumption that $\pi$ is model-defining. If $\phi$ is of the form $\psi_0 \lor \psi_1$, then $\pi(\phi)=0$ if and only if $\pi(\psi_0)=0=\pi(\psi_1)$ if and only if $\pi(\neg \psi_0)\neq 0\neq \pi(\neg \psi_1)$ if and only if $\pi(\neg \phi)\neq 0$. The first and third ``if and only if'' follow by positivity of $K$, and the second by induction assumption. The remaining cases are analogous.
\end{proof}

Two sentences $\phi$ and $\psi$ are \emph{$K$-equivalent}, written $\phi \equiv_K \psi$,  if $\evalX{\phi}{\pi}=\evalX{\psi}{\pi}$ for all {\color{black}model-defining} $K$-interpretations $\pi$. The sentences $\phi$ and $\psi$ are \emph{equivalent},  
written $\phi \equiv \psi$, 
if they are $K$-equivalent for all semirings 
 $K$.
Two logics $\mathfrak{L}$ and $\mathfrak{L}'$ are \emph{equally expressive under $K$} (resp. \emph{equally expressive}), denoted $\mathfrak{L}\equiv_K \mathfrak{L}'$ (resp. $\mathfrak{L}\equiv \mathfrak{L}'$), if all sentences from $\mathfrak{L}$ are $K$-equivalent (resp. equivalent) to some sentence from $\mathfrak{L}'$, and conversely all sentences from $\mathfrak{L}'$ are $K$-equivalent (resp. equivalent) to some sentence from $\mathfrak{L}$.

The following is a consequence of Proposition~\ref{prop:posmodeldef}.

\begin{proposition}\label{prop:caneq}
If $\phi$ and $\psi$ are $\FO$-formulae with formula (in)equalities, then
$\phi \leq \psi\equiv_{\SB}\phi \to \psi$ and $\phi = \psi \equiv_{\SB} \phi \leftrightarrow \psi$. 
\end{proposition}
\begin{corollary}
$\fofull\equiv_{\SB}\FO$.
\end{corollary}

\subsection{$K$-atoms}\label{sec:Katoms}
We are now ready to explore the idea of using logical statements as definitions of dependencies across various semirings.
To do so, we will consider an atom $\alpha$, like the dependence or independence atom, and define its interpretation over $K$-teams by referencing a definition of $\alpha$ stated in first-order logic with formula (in)equalities.

Consider a relation symbol $R$ that does not belong to $\tau$ (and can be of any arity).
For a tuple $\vec{a}=(a_1, \dots ,a_{\ar(R)})$ and a tuple of indices $\vec{i}=(i_1,\dots ,i_k)$ from $1,\dots,\ar(R)$, 
we define $\vec{a}_{\vec{i}} \dfn (a_{i_1}, \dots ,a_{i_k})$.
For tuples of indices $\vec{i}_1,\dots ,\vec{i}_n$ from $1,\dots,\ar(R)$ and variable tuples $\vec{u}_1, \dots ,\vec{u}_n$ (such that the length of $\vec{u}_l$  is that of $\vec{i}_l$, for each $l\leq n$), we define a shorthand
\[
\theta_{\vec{i}_1, \dots ,\vec{i}_n}(\vec{u}_1, \dots ,\vec{u}_n) \dfn \exists \vec{x}(R(\vec{x}) \land \vec{x}_{\vec{i}_1}=\vec{u}_1\land \dots \land\vec{x}_{\vec{i}_n}=\vec{u}_n).
\]
This shorthand formula expresses that there exists an $R$-fact such that its projections on sequences of positions
$\vec{i}_1, \dots ,\vec{i}_n$ are $\vec{u}_1, \dots ,\vec{u}_n$.
Considering dependence, independence, and inclusion atoms, we now define the following sentences:
\begin{align*}
\phi^{\vec{i}}_{\lite{S}} \dfn& \, \forall \vec{x} \big( R(\vec{x}) = \bot \lor \big (R(\vec{x})\neq \bot \land S(\vec{x}_{\vec{i}}) \big)\\
%
\phi^{\vec{i},\vec{j}}_{\rm dep} \dfn& \,
\forall\vec{u}\vec{v}\vec{w}\Big(\big(\theta_{\vec{i},\vec{j}}(\vec{u},\vec{v})\land \theta_{\vec{i},\vec{j}}(\vec{u},\vec{w})\big) = \bot 
 \lor (
\vec{v}=\vec{w} ) \neq \bot \Big) \\
%
\phi^{\vec{i},\vec{j},\vec{k}}_{\rm indep} \dfn& \, \forall {\vec{u}}\vec{v}\vec{w}\Big(\big(\theta_{\vec{i},\vec{j}}(\vec{u},\vec{v})\land \theta_{\vec{i},\vec{k}}(\vec{u},\vec{w})\big)\\
&\hspace{3cm}=\big(\theta_{\vec{i}}(\vec{u})\land\theta_{\vec{i},\vec{j},\vec{k}}(\vec{u},\vec{v},\vec{w})\big)\Big)\\
\phi^{\vec{i},\vec{j}}_{\rm inc} \dfn& \, \forall \vec{u} \Big(\theta_{\vec{i}}(\vec{u}) \leq \theta_{\vec{j}}(\vec{u})\Big).
\end{align*}


In the superscript, we may replace each unary tuple $(i)$ with $i$; e.g., write $\phi^{i}_{\lite{S}}$ instead of $\phi^{(i)}_{\lite{S}}$.
The above formulae can often be simplified, as illustrated next.
\begin{example}\label{ex:depstatements}
If $R$ is ternary then $\phi^{1,2,3}_{\rm indep}$ is of the form
 $\forall uvw\Big(\big(\theta_{1,2}({u},{v})\land \theta_{1,3}({u},{w})\big)
=\big(\theta_{{1}}({u})\land\theta_{{1},{2},{3}}({u},{v},{w})\big)\Big)$,
where $\theta_{1,2}(u,v)$, $\theta_{1,3}(u,w)$, $\theta_{1}(u,w)$, and $\theta_{1,3}(u,w)$ are respectively of the form
\begin{align*}
&\exists x_1x_2x_3  (R(x_1,x_2,x_3)\land x_1 ={u} \land x_2 ={v}),\\
&\exists x_1x_2x_3  (R(x_1,x_2,x_3)\land x_1 ={u} \land x_3=w),\\
&\exists x_1x_2x_3  (R(x_1,x_2,x_3)\land x_1 ={u}),\\
&\exists x_1x_2x_3  (R(x_1,x_2,x_3)\land x_1 ={u} \land x_2 ={v}\land x_3=w).
\end{align*}
Clearly, these sentences are equivalent to the simpler forms $
\exists x_3  R(u,v,x_3)$, 
$\exists x_2  R(u,x_2,w)$, 
$\exists x_2x_3  R(u,x_2,x_3),$ and $
R(u,v,w)
$, respectively.
 These equivalences can then be used to rewrite $\phi^{1,2,3}_{\rm indep}$ more succinctly.
\end{example}

It can now be observed that $\phi_{\rm dep}$, $\phi_{\rm indep}$ and $\phi_{\rm inc}$ are $\SB$-equivalent to the standard relational definitions of dependence, independence, and inclusion atoms. The following proposition is a consequence of Proposition \ref{prop:caneq}. It can be proven by imitating the reasoning in Example \ref{ex:logicaldef}.
\begin{proposition}\label{prop:fodefs}
The following equivalences hold:
\begin{align*}
\phi^{\vec{i}}_{\lite{S}} &\equiv_{\SB} \forall \vec{x} \big( R(\vec{x}) \rightarrow S(\vec{x}_i) \big) \\
\phi^{\vec{i},\vec{j}}_{\rm dep} &\equiv_{\SB} \forall\vec{u}\vec{v}\vec{w}(\theta_{\vec{i},\vec{j}}(\vec{u},\vec{v})\land \theta_{\vec{i},\vec{j}}(\vec{u},\vec{w})\to \vec{v}=\vec{w}) \\
\phi^{\vec{i},\vec{j},\vec{k}}_{\rm indep} &\equiv_{\SB} \forall {\vec{u}}\vec{v}\vec{w}(\theta_{\vec{i},\vec{j}}(\vec{u},\vec{v})\land \theta_{\vec{i},\vec{k}}(\vec{u},\vec{w})\to \theta_{\vec{i},\vec{j},\vec{k}}(\vec{u},\vec{v},\vec{w})) \\ 
\phi^{\vec{i},\vec{j}}_{\rm inc} &\equiv_{\SB} \forall \vec{u} (\theta_{\vec{i}}(\vec{u}) \to \theta_{\vec{j}}(\vec{u}))
\end{align*}
\end{proposition}

Having formalised key dependency concepts using logical statements, 
let us then move on to $K$-teams.
Now, fix a total order $<$ on the variable set $\Var$.
Let $\KX \colon \as(D,A)\to K$ be a $K$-team with domain $V=\{x_1, \dots ,x_k\}$, where $x_1 < \dots < x_k$. 
For each assignment $s \colon A^V \to K$, define a tuple $\vec{a}_s \dfn (s(x_1), \dots ,s(x_k))$.
Let $R$ be a  relation symbol of arity $k$. 
 Denote by $\pi_{\KX}:\lit{A}\to K$ any 
 $K$-interpretation such that
   $\pi_\KX$ maps $R(\vec{a}_s)$ to $\KX(s)$. 
 For a tuple of variables $\vec{x}=(x_{i_1}, \dots ,x_{i_n})$, write $\vec{i}_{\vec{x}}$ for the integer tuple $(i_1, \dots ,i_n)$. The $K$-interpretation of literals and dependencies is now defined as follows:
\begin{itemize}
\item Literals: $\evalX{T(\vec{x})}{\KX} \dfn \evalX{\phi^{\vec{i}_{\vec{x}}}_{\lite{T}}}{\pi_{\KX}}$.
\item Dependence atoms: $\evalX{\dep{\vec{x},\vec{y}}}{\KX} \dfn \evalX{\phi^{\vec{i}_{\vec{x}},\vec{i}_{\vec{y}}}_{\rm dep}}{\pi_{\KX}}$.
\item Independence atom: $\evalX{\pci{\vec{x}}{\vec{y}}{\vec{z}}}{\KX} \dfn \evalX{\phi^{\vec{i}_{\vec{x}},\vec{i}_{\vec{y}},\vec{i}_{\vec{z}}}_{\rm indep}}{\pi_{\KX}}$.
\item Inclusion atom: $\evalX{\vec{x}\subseteq\vec{y}}{\KX} \dfn \evalX{\phi^{\vec{i}_{\vec{x}},\vec{i}_{\vec{y}}}_{\rm inc}}{\pi_{\KX}}$.
\end{itemize}
\vspace{1mm}
We say that a $K$-team $\KX$ \emph{satisfies} an atom $\alpha$, denoted by $\KX\models \alpha$, if $\evalX{\alpha}{\KX}\neq 0$.

For instance, independence atom for the probability semiring corresponds to the notion of conditional independence in probability theory, and for the Boolean semiring it corresponds to the notion of \emph{embedded multivalued dependency} in database theory.

\begin{example}\label{ex:semantics}
Consider a pure independence atom $\pmi{x}{y}$ for the three $K$-teams presented in Fig. \ref{tab:three}. This independence atom is interpreted
in $K$-teams using the sentence $\phi^{\emptyset, 1,2}_{\rm indep}$. Analogous to Example \ref{ex:depstatements}, we may rewrite this sentence
in a simpler form:
\begin{equation}\label{eq:example}
 \forall uv\Big(\big(\exists y R(u,y) \land \exists x R(x,v)\big)
=\big(\exists xy R(x,y) \land R(u,v)\big)\Big)
\end{equation}
Suppose $x<y$ according to the total order $<$ on variables. Considering the $\SB$-team $\KX_1$, the function $\pi_{\KX_1}$  maps facts $R(a,a)$ and $R(a,b)$ to $1$, and facts $R(b,a)$ and $R(b,b)$ to $0$. 
 Then, $\pi_{\KX_1}$ interprets the formula equality in \eqref{eq:example} as $(1\land 1) = (1\land 1)$ for $(u,v)\mapsto \{(a,a),(a,b)\}$,
 and as $(0\land 1) = (1\land 0)$ for $(u,v)\mapsto \{(b,a),(b,b)\}$. Hence $\evalX{\pmi{x}{y}}{\KX_1}=1$, meaning that $\KX_1\models \pmi{x}{y}$. An alternative way to obtain $\KX_1\models\pmi{x}{y}$ is to use Propositions \ref{prop:can} and \ref{prop:fodefs}, noting that the model (over signature $\{R\}$) defined by $\KX_1$ satisfies a first-order sentence that is $\SB$-equivalent to $\phi^{\emptyset, 1,2}_{\rm indep}$.
Using similar calculations we may further observe $\KX_3$ satisfies $\pmi{x}{y}$ while $\KX_2$ does not. 
 On the other hand, $\KX_2$ is the only $K$-team of the three satisfying the dependence atom $\dep{x,y}$.
\end{example}

Recall that $\cmap{K}:K\to \SB$ is the characteristic mapping that associates non-zero values of $K$ with $1$ and zero with $0$. The following proposition shows that this mapping preserves the truth of all $\fopos$-formulae.
\begin{proposition}
Let $\pi$ be a $K$-interpretation over a positive semiring $K$. 
Then for all $\fopos$-definable $\phi$,  $\evalX{\phi}{\cmap{K}\circ \pi}= 0$ implies 
 $\evalX{\phi}{\pi}= 0$.
\end{proposition}
\begin{proof}
The proof proceeds by structural induction on the structure of $\phi$.
We prove simultaneously that the implication can be strengthened to if and only if when $\phi\in \FO$.
The cases for first-order literals follow directly from the definition of $\cmap{K}$, and the case for $\neg \phi$ is trivial. 

The cases for formula equalities and inequalities follow from the positiveness of $K$ together with the induction hypothesis. Below $\phi, \psi\in \FO$, and xor is the exclusive or.
\begin{align*}
\evalX{\phi = \psi}{\cmap{K}\circ \pi}= 0 &\quad\Leftrightarrow\quad \evalX{\phi}{\cmap{K}\circ \pi} = 0 \text{ xor } \evalX{\psi}{\cmap{K}\circ \pi}= 0\\
&\quad\Leftrightarrow\quad \evalX{\phi}{\pi}= 0 \text{ xor } \evalX{\psi}{\pi}= 0\\
&\quad\Rightarrow\quad \evalX{\phi = \psi}{\pi}= 0\\[5pt]
\evalX{\phi \leq \psi}{\cmap{K}\circ \pi}= 0 &\quad\Leftrightarrow\quad \evalX{\phi}{\cmap{K}\circ \pi} \neq 0 \text{ and } \evalX{\psi}{\cmap{K}\circ \pi} = 0\\
&\quad\Leftrightarrow\quad \evalX{\phi}{\pi} \neq 0 \text{ and } \evalX{\psi}{\pi } = 0\\
&\quad\Rightarrow\quad \evalX{\phi \leq \psi}{\pi}= 0
\end{align*}
The case for $\evalX{\phi \neq \bot}{\cmap{K}\circ \pi}= 0$ is similar.

The cases for $\land$, $\lor$, $\exists$, and $\forall$ follow from the positiveness of $K$ together with the induction hypothesis, we show $\lor$:
\begin{align*}
\evalX{\phi \lor \psi}{\cmap{K}\circ \pi}= 0 &\quad\Leftrightarrow\quad \evalX{\phi}{\cmap{K}\circ \pi}= 0 \text{ and } \evalX{\psi}{\cmap{K}\circ \pi}= 0\\
&\quad\Rightarrow\quad \evalX{\phi}{\pi}= 0 \text{ and } \evalX{\psi}{\pi}= 0\\
&\quad\Leftrightarrow\quad \evalX{\phi \lor \psi}{\pi}= 0
\end{align*}
The implication above follows from the induction hypothesis and can be strengthened to an equivalence if $\phi, \psi\in \FO$.
\end{proof}
Let us now define the \emph{possibilistic collapse}  of a $K$-team $\KX$ is as the $\SB$-team $\cmap{K} \circ \KX$. We sometimes identify the possibilistic collapse with the team it defines (i.e., the set of assignments it maps to $1$).
As a consequence of the previous proposition, any $\fopos$-definable atom is preserved under the possibilistic collapse.
\begin{corollary}\label{cor:collapse}
Let $\KX$ be a $K$-team over a positive semiring $K$, and let $\alpha$ be an $\fopos$-definable atom. If $\KX$ satisfies $\alpha$, then its possibilistic collapse satisfies $\alpha$. The converse direction holds true if 
$\alpha$ is $\foneg$-definable.
\end{corollary}
For instance, in Example \ref{ex:semantics} we observed that the $\SB$-team $\KX_1$ satisfies the independence atom $\pmi{x}{y}$. Since $\KX_1$ is the possibilistic collapse of $\KX_3$, this follows already by the fact that $\KX_3$ satisfies the same atom, which in turn is $\fopos$-definable.
%
We also noted that the $\mathbb{N}$-team $\KX_2$ in Fig. \ref{tab:three} satisfies the dependence atom $\dep{x,y}$. By Corollary \ref{cor:collapse} this follows also from the fact that the possibilistic collapse of $\KX_2$ satisfies the same $\foneg$-definable atom.

At this point we note that an alternative way to interpret the dependence atom $\dep{x,y}$ in a data table with duplicates (i.e., in an $\mathbb{N}$-team), would be to stipulate that $x$ uniquely determines $y$ and has no duplicates in the projection of the table to $x$ and $y$ (this would be a natural extension of the semiring interpretation of keys by \cite[Definition 4.1]{ChuMRCS18}). However, such a notion of dependence would
fail to satisfy the reflexivity rule of functional dependencies, which entails that $\dep{x,x}$ always holds.

 \section{Team semantics}\label{sect:team-semantics} 
Next we explore how all the team semantics variants (and more) can be unified under the rubric of semirings.
We first present an adaptation of team semantics for $K$-teams, and then consider $K$-interpretations of complex formulae.

For the notion of team semantics, a few useful concepts are needed.
The \emph{projection} $\proj{X}{V}$ of $X$ on a variable set $V\subseteq \Dom(X)$ is defined as the set $\{\proj{s}{V}\mid s \in X\}$, where $\proj{s}{V}$ is the projection of an assignment $s$ on $V$ defined in the usual way. For a set $S$ and a variable, we define $X[S/x]$ as the team $\{s[a/x]\mid s\in X, a\in S\}$. For a set $S$ and a function $F: X \to \Pow(S)\setminus \{\emptyset\}$, we define $X[F/x]$ as the team $\{s[a/x]\mid s\in X, a\in F(s)\}$. 

We present the standard team semantics for $\FO$ in negation normal form. 

\begin{definition}[Team semantics]\label{def:teamsemantics}
Let $X$ be a team of a first-order structure $\mA$ over vocabulary $\tau$. For $\phi \in \FO[\tau]$, we define when $X$ \emph{satisfies} $\phi$ \emph{under} $\mA$, written $\mA\models_X \phi$, as follows ($\models_s$ refers to the usual Tarski semantics of FO):
	\begin{tabbing}
		\= $\mA \models_{X} l$ \hspace{7ex}\= $\Leftrightarrow$\quad \= $\mA\models_s l$ for all $s\in X$ ($l$ is a literal),\\ 
		\> $\mA\models_{X} (\psi \land \theta)$ \> $\Leftrightarrow$ \> $\mA \models_{X} \psi \text{ and } \mA \models_{X} \theta$,\\
		\> $\mA\models_{X} (\psi \lor \theta)$ \> $\Leftrightarrow$ \> $\mA\models_{Y} \psi \text{ and } \mA \models_Z \theta$\text{ for some }\\
		\>\>\>\text{$Y,Z\sub X$ such that }$Y\cup Z = X$,\\
		\> $\mA\models_{X} \forall x\psi$ \> $\Leftrightarrow$ \> $\mA\models_{X[A/x]} \psi$,\\
		\> $\mA\models_{X} \exists x\psi$ \> $\Leftrightarrow$ \> $\mA\models_{X[F/x]} \psi$\text{ for some function }\\
		\>\>\>$F: X \to \Pow(A)\setminus \{\emptyset\}$.
	\end{tabbing}
\end{definition}

\subsection{$K$-team semantics}

 By re-examining team semantics through the lens of semirings, 
we can arrive at the following truth definition.

\begin{definition}[$K$-team semantics]\label{def:kteamsemantics}
Let $\KX$ be a $K$-team of a first-order structure $\mA$ over vocabulary $\tau$. For $\phi \in \FO[\tau]$, we define when $\KX$ \emph{satisfies} $\phi$ \emph{under} $\mA$, written $\mA\models_{\KX} \phi$:
	\begin{tabbing}
		\= $\mA \models_{\KX} l$ \hspace{7ex}\= $\Leftrightarrow$\quad \= $\mA\models_s l$ for all $s\in \support(\KX)$ ($l$ is a literal),\\ 
		\> $\mA\models_{\KX} (\psi \land \theta)$ \> $\Leftrightarrow$ \> $\mA \models_{\KX} \psi \text{ and } \mA \models_{\KX} \theta$,\\
		\> $\mA\models_{\KX} (\psi \lor \theta)$ \> $\Leftrightarrow$ \> $\mA\models_{\KY} \psi \text{ and } \mA \models_{\KZ} \theta \text{ for some }\KY,\KZ$\\
		\>\>\>$\text{such that } \forall s: \KY(s) + \KZ(s) = \KX(s)$,\\
		\> $\mA\models_{\KX} \forall x\psi$ \> $\Leftrightarrow$ \> $\mA\models_{\KY} \psi$, where $\KY$ is such that \\
		\>\>\>$\forall s,a:\KX(s) = \KY(s[a/x])$\\
		\> $\mA\models_{\KX} \exists x\psi$ \> $\Leftrightarrow$ \> $\mA\models_{\KY} \psi$\text{ for some $\KY$ such that }\\
		\>\>\>$\forall s:\KX(s) = \sum_{a} \KY(s[a/x])$.
	\end{tabbing}
\end{definition}

For the Boolean semiring, the above definition gives the standard team semantics presented in Definition \ref{def:teamsemantics}. For the semiring of natural numbers, we obtain multiteam semantics \cite{GradelW22}, and for the probability semiring we obtain probabilistic team semantics \cite{abs-2003-00644}.

The extension of first-order logic with dependence atoms is called 
\emph{dependence logic}. Similarly \emph{independence logic} and \emph{inclusion logic}
 are the extensions of $\FO$ with conditional independence atoms
 and inclusion atoms, respectively. The interpretations of relational and dependency atoms are as defined in Section \ref{sec:Katoms},
 except that the definitions are of the form $\evalX{-}{\mA,\KX} \dfn \evalX{-}{\pi_{\mA,\KX}}$ (instead of $\evalX{-}{\KX} \dfn \evalX{-}{\pi_{\KX}}$), where $\pi_{\mA,\KX}$ is a model defining interpretation for $\mA$ that encodes both $\mA$ and $\KX$.
 We then stipulate $\mA\models_{\KX} \alpha$, if $\evalX{\alpha}{\mA,\KX}\neq 0$, when $\alpha$ is an atom or a literal.
 Note that the case for literals given in the above definition coincides with the definition of Section \ref{sec:Katoms}.

It has been observed that seminal ``No-Go'' theorems in quantum mechanics, such at the Bell's theorem or the Kochen-Specker theorem, can be formalised as a logical entailment $\Sigma \models \phi$, where $\Sigma \cup\{\phi\}$ is a collection of  dependence or independence logic formulae \cite{AlbertG22,abramsky2021team}. In this particular context, it does not make any difference whether one considers relational or probabilistic team semantics. 
Also in general the two are connected: for independence logic, satisfaction in probabilistic team semantics implies satisfaction in relational team semantics, and the converse holds for dependence logic \cite[Theorem 3.5]{AlbertG22} and \cite{DurandHKMV18}. Using $K$-teams, these results can now be stated in the following more general form.

Given a collection of atoms $C$, we write $\FO(C)$ for the extension of (negation normal form) $\FO$ with atoms in $C$.
We say that a semiring $K$ is \emph{$+$-dense} if for all nonzero $a\in K$ there exist nonzero $b,c\in K$ such that $a=b+c$. 
\begin{theorem}
Let $C$ and $D$ be collections of $\fopos$-definable and $\foneg$-definable atoms, resp. Assume $\phi\in\FO(C)$ and $\psi\in\FO(D)$. 
Let $\mA$ be a first-order structure, $\KX$ a $K$-team over a positive semiring $K$, and
$X$ the possibilistic collapse of $\KX$. Then, $\mA\models_{\KX} \phi \Rightarrow \mA \models_X \phi$. Moreover, if $K$ is $+$-dense, then $\mA\models_{\KX} \psi \Leftarrow \mA \models_X \psi$.
\end{theorem}
\begin{proof}
The proof proceeds by structural induction on the formulae; Corollary \ref{cor:collapse} is the atomic case. The case for conjunction is trivial, and the cases for $\lor$, $\exists$, and $\forall$ are similar to each other. We show the case of $\lor$. Consider the following:
\begin{align*}
&\mA\models_{\KY} \theta_1 \text{ and } \mA \models_{\KZ} \theta_2 \text{ for } \KY,\KZ
\text{ s.t. } \forall s: \KY(s) + \KZ(s) = \KX(s)\\
&\mA\models_{Y} \theta_1 \text{ and } \mA \models_{Z} \theta_2 \text{ for some } Y, Z
\text{ such that } Y \cup Z = X
\end{align*}
The first line is by definition equivalent to $\mA\models_{\KX} \theta_1 \lor \theta_2$, and the second line to $\mA\models_{X} \theta_1 \lor \theta_2$. Assuming the first line, we obtain the second line with $Y$ and $Z$ being supports of $\KY$ and $\KZ$ by $+$-positiveness. Assuming the second line, we obtain the first line with $Y$ and $Z$ being supports of $\KY$ and $\KZ$ by $+$-denseness.
\end{proof}

\subsection{Algebraic $K$-team semantics}\label{sec:algebraic_semantics}
Next we reformulate $K$-team semantics via constrained polynomials  $\eval{\phi}{\mA}{\KX}$ over $K$ which can be used for provenance analysis and various counting tasks.
Due to the use of identities, constrained polynomials are terms over the expansion of $(K,+,\cdot,0,1)$ by suitable $\chi$-functions giving access to identity between terms.
Our approach can be used to reduce satisfaction in $K$-team semantics to the existential first-order theory of $K$. 
Similar reductions have been utilised in the case of Boolean and probabilistic team semantics to analyse the complexity of model checking and satisfiability \cite{jelia19,abs-2003-00644,HannulaV22,DKV22} 

Let $\mA$ be a finite model with universe $A$. Let $V$ be a finite set of variables and $\vec{a}\in A^{V}$. Below $\KX(\vec{a})$ denotes a variable over $K$. Observe that by fixing interpretations for 
$\KX(\vec{a})$, for all  $\vec{a}\in A^{V}$, a unique $K$-team $\KX$ is  determined.  The constrained polynomial $\eval{\phi}{\mA}{\KX}$ contains also other variables $\KY(\vec{b})$ and $\KZ(\vec{c})$ that represent new teams that arise along the evaluation of disjunctions and the quantifiers, where $\vec{b}\in A^{V_1}$ and $\vec{c}\in A^{V_2}$, for $V_1,V_2 \supseteq V$. So $\eval{\phi}{\mA}{\KX}$ defines a function from $K^n$ to $K$ (where $n$ is the number of variables of $\eval{\phi}{\mA}{\KX}$), which yields a value in $K$ once values for all the free variables have been fixed.

\begin{definition}Let $\mA$ be a finite model and $V$ a finite set of variables.
We define $K$\emph{-interpretation} $\eval{\cdot}{\mA}{\KX}$ as follows. 
Below $\chi$ is the characteristic function of equality (with respect to $0$ and $1$ from $K$), and $a$ and $\vec{a}$ range over $A$ and tuples from $A$, resp. Note that each tuple from $\vec{a}\in A^V$ gives rise to an assignment $s\colon V \to A$ such that $\vec{a} = \big( s(x_1),\dots, s(x_n) \big)= \vec{a}_s$.
\begin{align*}
    \eval{\phi \lor \psi}{\mA}{\KX} &= \eval{\phi }{\mA}{\KY} \cdot     \eval{\psi }{\mA}{\KZ}  \cdot \prod_{\vec{a}} \chi[\KY(\vec{a}) + \KZ(\vec{a}) = \KX(\vec{a})] \\   \quad     \eval{\phi \land \psi}{\mA}{\KX}& =     \eval{\phi }{\mA}{\KX} \cdot  \eval{\psi }{\mA}{\KX} \\
  \eval{\forall x\phi}{\mA}{\KX} & =    \eval{\phi }{\mA}{\KY} \cdot \prod_{a,s\colon \Dom(\KX)\to A} \chi[\KX(\vec{a}) = \KY(\vec{a}_{s[a/x]})]  &\quad \\
     \eval{\exists x\phi}{\mA}{\KX} & = \eval{\phi }{\mA}{\KY} \cdot \prod_{s\colon \Dom(\KX)\to A} \chi[\KX(\vec{a}) = \sum_{a} \KY(\vec{a}_{s[a/x]})]
\end{align*}
For first-order literals and atoms, we utilise the interpretations defined in Section \ref{sec:Katoms} with the modification discussed in the previous subsection. Recall that the $K$-interpretation of dependencies and relational atoms is defined in terms of $\pi_{\mA, \KX}:\lit{A}\to K$ mapping $R(\vec{a})$ to $\KX(\vec{a})$, where $R$ is a relation symbol (not in the vocabulary of $\mA$) representing the team  $\KX$. We further assume that, except for $R$, $\pi_{\mA, \KX}$ is identical to the canonical truth interpretation $\pi_{\mA}$ of $\mA$.


Let $T(\vec{x})$ be a first-order literal and $\mA$ a structure. Then, $\evalX{T(\vec{x})}{\mA,\KX}$  can be expanded into:
\begin{align}
    \prod_{s\colon \Dom(\KX)\to A} \Big( \chi[\KX(\vec{a}_s)=0] + \chi[\KX(\vec{a}_s)\neq  0]\cdot  T\big(s(\vec{x})\big) \Big) \label{counting_poly}
\end{align}
where $\KX(\vec{a})$ and $T(\vec{a})$ are interpreted according to $\pi_{\mA,\KX}$.
Now, the definitions of literals and dependencies given in Section \ref{sec:Katoms} can be imported into the algebraic semantics by viewing strings of the form $\KX(\vec{a})$ as variables ranging over $K$. 
\end{definition}

It is straightforward to show that $K$-team semantics of Definition \ref{def:kteamsemantics} coincides in the following sense with algebraic $K$-team semantics.
\begin{proposition}
$\mA\models_{\KX}\phi$ iff $\Ran(\eval{\phi}{\mA}{\KX})\neq \{0\}$,
where $\Ran(\eval{\phi}{\mA}{\KX})$ is the range of the function defined by $\eval{\phi}{\mA}{\KX}$
 when the interpretations of $\KX(\vec{a})$ are fixed according to $\KX$.
\end{proposition}

Note that any constrained polynomial $\eval{\phi}{\mA}{\KX}$ can be defined by an $\FO$-formula over $K$, and thus checking  $\mA\models_{\KX}\phi$ can be reduced (in polynomial time) to the existential first-order theory of $K$ with additional constants for $\KX(\vec{a})$. 
E.g., satisfaction of a literal $T(\vec{x})$ can be expressed by the 
following formula $\psi$ over an expansion of $(K,+,\cdot,0,1)$ with additional constants from $K$:
\[
\psi \dfn \bigwedge_{\mA\not\models_s T(\vec{x})}\KX(\vec{a}_s)=0.
\]
Now $\evalX{T(\vec{x})}{\mA,\KX}  \neq 0$ iff $\big(K,+,\cdot,0,1, \KX(\vec{a}_1),\dots, \KX(\vec{a}_n)\big)\models \psi $.
It is worth noting that formalising $K$-team semantics of sentences can be done in existential first-order theory of $K$ without additional constants from $K$.


\section{Conclusions and future work}\label{sect:applications}

We defined an extension of $\FO$ under semiring semantics with the ability of comparing semiring values of first-order formulae.
We used this formalism to define concepts such as dependence and independence in a way that encompasses prior interpretations and indicated its advantages in studying the preservation of satisfaction and entailment for dependence statements between different semirings. Such preservation results have previously been studied between database and probability theory. 
We proposed a unifying approach inspired by semiring provenance for analysing the concepts of dependence and independence via a novel semiring team semantics, which subsumes all the previously considered variants for first-order team semantics.
We discovered general explanations for the preservation of satisfaction results from team-semantics literature. 
We conclude by exploring some applications and directions for future work.

\subsection{Axiomatisations and logical implication}
 The notions of dependence and independence  are known to exhibit remarkable similarity in their behavior across various contexts in which they are defined. One example of this are the \emph{Armstrong axioms} \cite{armstrong74}, which describe the laws of inference for functional dependence in relational databases. In this context, if every two tuples in a database that agree on an attribute set $X$ also agree on an attribute set $Y$, we say that $Y$ \emph{functionally depends} on $X$. The Armstrong axioms seem to capture something more fundamental and universal than just this concept.
 For instance, if we consider Shannon's information measures, we can say that a random variable $Y$ depends functionally on another random variable $X$ whenever the conditional entropy $H(Y\mid X)$ of $Y$ given $X$ equals $0$. Similarly, in linear algebra, we may say that 
 a subspace $Y$ of a vector space $V$ depends functionally on another subspace $X$ of $V$ if every vector of $Y$ is a linear combination of vectors in $X$. In all these cases, and in many others, the Armstrong axioms are sound and complete  (see, e.g., \cite{GallianiV22}).

 When it comes to the notion of independence, there are similarities but also differences. As for the similarities, the axioms of \emph{marginal independence} (here, pure independence)  $\pci{\emptyset}{X}{Y}$ formulated by \cite{geiger:1991} in the context of probability theory, are known to be sound and complete in the database context \cite{KontinenLV13}. 
 This correspondence between logical implication in probability theory and database theory extends to the so-called \emph{saturated conditional independence} (in databases, \emph{multivalued dependency}) $\pci{X}{Y}{Z}$, where $X\cup Y\cup Z$ has to cover all variables of the joint distribution (in databases, all attributes of the relation schema) \cite{WongBW00},  as well as their extension with functional dependencies \cite{KenigS22}. 
 Logical implication for the general conditional independence (in databases, {embedded multivalued dependencies}) however is not the same for probability distributions and database relations \cite{studeny:1993}.

 It is noteworthy that this connection between database theory and probability theory seems to hold as long as there exists a common foundation through information theory. Indeed,
  marginal independence, saturated conditional independence, and functional dependence can in both contexts be interpreted through information-theoretic measures \cite{Lee87,GallianiV22}. On the other hand, there does not seem to exist any evident  information-theoretic interpretation for the embedded multivalued dependency of database theory. The semiring approach proposed in this paper manages to unify dependency concepts from various contexts; in particular, conditional independence from probability theory and database theory.
  In doing so, it offers the potential to shed new light on the underlying reasons behind said similarities and differences. 
\begin{wraptable}{R}{0.17\textwidth}
\caption{Mixing fails}
\label{tab:z4}
\centering
\begin{tabular}{@{}ccc:c@{}}
\multicolumn{4}{c}{$ K = \mathbb{Z}_4$} \\
\toprule
$x$ & $y$ &$z$& $\KX(s)$ \\
\midrule
 $a_0$ &$b_0$ & $c_0$ & 1 \\
 $a_0$ &$b_1$ & $c_0$ & 1 \\
 $a_1$ &$b_0$ & $c_1$ & 1 \\
 $a_1$ &$b_1$ & $c_1$ & 1 \\
 $a_2$ &$b_2$ & $c_0$ & 1 \\
 $a_2$ &$b_3$ & $c_0$ & 1 \\
 $a_3$ &$b_2$ & $c_1$ & 1 \\
 $a_3$ &$b_3$ & $c_1$ & 1 \\
\bottomrule
\end{tabular}
\end{wraptable} 

 To illustrate what this sort of semiring approach might reveal, we provide an example that shows how the axiomatic properties of independence may sometimes hinge on the underlying algebraic properties.
\begin{example}
    An element $a$ of a (commutative) semiring $K=(K,+,\cdot,0,1)$ is said to be 
    \emph{cancellative} if for all $b,c\in K$, $ab=ac$ implies $b=c$. 
It can be shown that the axioms of pure independence are sound for $K$-teams if every element $a \in K\setminus \{0\}$ is cancellative. If this condition fails, the \emph{mixing rule} \cite{geiger:1991} of pure independence is not necessarily sound. This rule states that $\pmi{x}{yz}$ can be derived from $\pmi{x}{y}$ and $\pmi{xy}{z}$. For a counterexample, 
%
 the ring $\mathbb{Z}_4$ of integers modulo $4$ contains a non-cancellative element $2\neq 0$. 
If we define a $\mathbb{Z}_4$-team $\KX$ as in Table \ref{tab:z4}, we observe that $\KX$ satisfies $\pmi{x}{y}$ and $\pmi{xy}{z}$, but  fails to satisfy $\pmi{x}{yz}$.
\end{example}

\subsection{Provenance and counting proofs}

The introduction of the algebraic semantics is partially motivated by the fact that $\eval{\phi}{\mA}{\KX}$ can be used for provenance analysis and counting tasks. 
In provenance information is extracted from \emph{tokens} (or \emph{annotations}). In the $K$-team setting,  each assignment is annotated with a token. Tokens are used to trace the origin of the truth value of a given formula by
interpreting an expression of some sort. Our goal is to understand how a formula ends up being true in a first-order structure with $K$-team semantics.

Let $K$ be a commutative positive semiring. If a formula $\phi$ is true in a non-empty $K$-team, then we would like to obtain a polynomial expression involving the semiring values given to each assignment of the $K$-team that explains the truth of $\phi$. Instead, if the given formula is false we would like such expression to return $0$.

Notice that we already obtained a polynomial expression in Section \ref{sec:algebraic_semantics}, where the annotations played a role in the definition of algebraic $K$-team semantics. However, the literals' truth values lacked annotations. 
For a concrete $K$-team $\KX$, the number of different ways of satisfying a formula $\phi$ over $\mA$ and $\KX$ corresponds to the cardinality of the support of $\eval{\phi}{\mA}{\KX}$ (i.e., the number of assignments with domain $\Dom(\eval{\phi}{\mA}{\KX})\setminus \Dom(\KX)$ such that the expression returns a nonzero value).
Moreover, $\eval{\phi}{\mA}{\KX}$ can also be devised for counting the number of $K$-teams $\KX$ that satisfy $\phi$ over $\mA$ (cf. \cite{HaakKMVY19}).

To trace provenance, we define the following sentence for literals in a similar manner as in Section \ref{sec:Katoms}:
\[
\phi^{\vec{i}}_{\liteprov{T}} \dfn \, \forall \vec{x} \big( R(\vec{x}) = \bot \lor \big (R(\vec{x}) \land T(\vec{x}_{\vec{i}}) \big)
\]

Then, we define $K$-team provenance semantics using an analogous interpretation for first-order literals and atoms as the one defined in Section \ref{sec:Katoms}. 
If $T(\vec{x})$ is a first-order literal, then $\evalX{T(\vec{x})}{\KX}$ can be expanded into:
 \[ \prod_{s\colon \Dom(\KX)\to A} \Big( \chi[\KX(\vec{a}_s)=0] + \KX(\vec{a}_s)\cdot T\big(s(\vec{x})\big) \Big).  
\] 
This is extended for general formulae as in Section \ref{sec:algebraic_semantics}.

\subsection{Repairs}\label{sec:repairs}

To transform a database to an accurate reflection of the domain it is intended to model, some properties and conditions are imposed on the possible instances to avoid inconsistency. A notion of consistency of the database is then related to a set of ICs, which express some of the semantic structure that the data intends (or needs) to represent. 
It is common for a database to become inconsistent due to several reasons.
When a database does not satisfy its ICs, one possible approach is to perform minimal changes to obtain a ``similar'' database that satisfies the constraints. Such a database is called a \emph{repair} \cite{ArenasBC99}, and to define it properly one has to precisely determine the meaning of ``minimal change''.
Several definitions have been proposed and studied in the literature, usually given in terms of a distance or partial order between database instances. Which notion to use may depend on the application. 

We use $K$-team semantics to determine whether a given $K$-team $\KX$ satisfies a set of ICs. Assuming we have a way to measure distances between $K$-teams, if the ICs are not satisfied we could ask for a \emph{repair} of $\KX$ that does. 
That is, a $K$-team $\KY$ such that the ICs are satisfied in $\KY$, and $\KX$ and $\KY$ minimally differ in terms of the desired distance.
In what follows, we restrict to ordered semirings and stipulate the existence of additive inverses (i.e., ordered rings). 

Since $K$-team semantics allows to define dependencies in $K$-teams, one could ask for a notion of a repair that takes into account either dependence or independence. One possibility is to define a quantitative notion of non-independence to a $K$-team by assigning a value in the semiring using the weights of the assignments, indicating how far away we are from having independence. 
When looking at the independence atom defined in Section \ref{sec:Katoms}, we interpret that we have independence between $\vec{x}$ and $\vec{y}$ in a $K$-team $\KX$ if the equality 

$\sum_{s} \KX(s) \cdot \sum_{s(\vec{x}\vec{y})=\vec{a}\vec{b}} \KX(s)= \sum_{s(\vec{x})=\vec{a}}\KX(s)\cdot \sum_{s(\vec{y})=\vec{b}}\KX(s)$ 

\noindent holds for every pair $\vec{a},\vec{b}$.
If instead $\evalX{\pmi{x}{y}}{\KX}=0$, then at least one of these terms is false. Hence, for every pair $\vec{a},\vec{b}$ for which the equality does not hold, we measure how far away they are from being equal. More precisely, we consider:
\[ \evalX{\nonpmi{\vec{x}}{\vec{y}}}{\KX} =  \sum_{\vec{a},\vec{b}} \lvert \sum_{s} \KX(s) \cdot \sum_{s(\vec{x}\vec{y})=\vec{a}\vec{b}} \KX(s) - \sum_{s(\vec{x})=\vec{a}}\KX(s)\cdot \sum_{s(\vec{y})=\vec{b}}\X(s) \rvert \]
where the module $\lvert a-b \rvert$ is defined as $a-b$ if $a-b>0$, and $b-a$ otherwise, for any $a, b \in K$.

We now present a natural way to define distance between $K$-teams using the values in the semiring, and then introduce some notions of $K$-team repairs. 

Let $\KX, \KY$ be two $K$-teams. We define the \emph{symmetric difference of $\KX$ and $\KY$}, 
denoted by $\simdif{\KX}{\KY}$, 
as the $K$-team with weights $(\simdif{\KX}{\KY})(s)$ defined as:
\[ (\simdif{\KX}{\KY})(s) =  \lvert \KX(s)-\KY(s) \rvert \] 
Using this, we define a distance between $\KX$ and $\KY$ as:
\[ dist(\KX, \KY) = \sum_{s} (\simdif{\KX}{\KY})(s) \]
Notice that, if we already have some kind of norm or distance in $K$ (or $K^k$), then we can consider instead said norm as a distance. 
 Moreover, if we have a distance in $K$, then we do not need to ask for additive inverses in $K$.

Some notions of repairs that arise naturally in this context are the following: given a $K$-team $\KX$ and a formula $\phi$,

\begin{itemize}
    \item 
      A \emph{symmetric difference repair} of $\KX$ w.r.t. $\phi$ is a $K$-team $\KY$ that satisfies $\phi$ and is such that $dist(\KX,\KY)\leq dist(\KX',\KY)$ for all $K$-teams $\KX'$ satisfying $\phi$.
    If $K$ is the Boolean semiring, this notion becomes the cardinality-based repair known as the \emph{C-repair} \cite{LopatenkoB07}.

    \item A \emph{subteam repair} (resp. \emph{superteam repair}) of $\KX$ w.r.t. $\phi$ is a $K$-team $\KY$ that is a subteam (resp. {superteam}) of $\KX$ satisfying $\phi$, and such that $dist(\KX,\KY)\leq dist(\KX',\KY)$ for all subteams (resp. {superteams}) of $\KX$ satisfying $\phi$.


    \item Assuming $\phi$ is of the form $\pmi{\vec{x}}{\vec{y}}$, we can also consider notions of repairs that minimise $\evalX{\nonpmi{\vec{x}}{\vec{y}}}{\KX\triangle \KY}$.
\end{itemize}

\subsection{Complexity and K-machines}

Similar to the way we generalised team semantics over semirings, there have been several approaches to do the same for computational complexity.
A prominent related model of computation 
is the so-called \emph{BSS-model} \cite{blum1989}.
A BSS-machine over a semiring $K$ can be thought of as a Turing machine which has a tape of $K$-valued registers instead of just zeros and ones. 
The transition function then allows evaluating polynomial functions 
on a fixed interval of the tape in a single step.
In their book \cite{BSSbook}, Blum, Cucker, Shub, and Smale predominantly use this model of computation to investigate questions in the realms of real and complex numbers and shed light on the differences between them, and the Turing model.
Indeed, changing the underlying semiring often leads to profound complexity theoretic implications.
Take the Hilbert's 10th problem for example,
the question whether a multivariate polynomial with integer coefficients has an integer solution is undecidable.
However, asking for real solutions leads the problem to become decidable. Moreover, it is an open problem whether there exists a general decision procedure to check the existence of rational solutions.





It is a fascinating avenue for future work to investigate what general results can be proven for our formalisms in the context of BSS-complexity. In the Boolean setting most team-based logics are known to characterise $\NP$ \cite{KV16} (and thus $\NP$ on BSS-machines with access to the Boolean semiring), while \cite{abs-2003-00644} show a corresponding characterisation between probabilistic independence logic and $\NP$ on a variant of BSS-machines with access to the probabilistic semiring.

\section*{Acknowledgments}
Miika Hannula has been supported by the ERC grant 101020762. Juha Kontinen was partially funded by Academy of Finland grant 338259.
Nina Pardal was supported by the DFG grant VI 1045-1/1.
Jonni Virtema was partially supported by the DFG grant VI 1045-1/1 and by Academy of Finland grant 338259.

\bibliographystyle{kr}
\bibliography{kr-sample,biblio}

\end{document}